\documentclass[10pt]{amsart}
\usepackage{epsfig,url}
\usepackage{mathdots}

\newtheorem{theorem}{Theorem}[section]
\newtheorem{lemma}[theorem]{Lemma}

\newtheorem{problem}[theorem]{Question}
\newtheorem{corollary}[theorem]{Corollary}
\newtheorem{proposition}[theorem]{Proposition}

\theoremstyle{definition}
\newtheorem{definition}[theorem]{Definition}
\newtheorem{example}[theorem]{Example}

\newtheorem{note}[theorem]{Note}

\theoremstyle{remark}

\begin{document}

\title[Sequences, polynomials and operator orderings]
{From sequences to polynomials and back, via operator orderings}

\author[]{Tewodros Amdeberhan}
\address{Department of Mathematics,
Tulane University, New Orleans, LA 70118}
\email{tamdeber@tulane.edu}

\author[]{Valerio De Angelis}
\address{Department of Mathematics,
Xavier University of Louisiana, New Orleans, LA 70125}
\email{vdeangel@xula.edu}

\author[]{Atul Dixit}
\address{Department of Mathematics,
Tulane University, New Orleans, LA 70118}
\email{adixit@tulane.edu}

\author[]{Victor H. Moll}
\address{Department of Mathematics,
Tulane University, New Orleans, LA 70118}
\email{vhm@tulane.edu}

\author[]{Christophe Vignat}
\address{Department of Mathematics,
Tulane University, New Orleans, LA 70118 and L.S.S. Supelec, Universite 
d'Orsay, France}
\email{vignat@tulane.edu}


\date{March 4, 2013}


\begin{abstract}
C. M. Bender and G. V. Dunne showed that linear combinations of 
words $q^{k}p^{n}q^{n-k}$, where $p$ and $q$ are subject to the 
relation $qp - pq = \imath$, may be expressed as a polynomial in the 
symbol $z = \tfrac{1}{2}(qp+pq)$. Relations between such  polynomials and 
linear combinations of the transformed coefficients
are explored. In particular, examples yielding orthogonal polynomials 
are provided. 
\end{abstract}

\maketitle

\newcommand{\ba}{\begin{eqnarray}}
\newcommand{\ea}{\end{eqnarray}}
\newcommand{\ift}{\int_{0}^{\infty}}
\newcommand{\nn}{\nonumber}
\newcommand{\no}{\noindent}
\newcommand{\lf}{\left\lfloor}
\newcommand{\rf}{\right\rfloor}
\newcommand{\realpart}{\mathop{\rm Re}\nolimits}
\newcommand{\imagpart}{\mathop{\rm Im}\nolimits}

\newcommand{\op}[1]{\ensuremath{\operatorname{#1}}}
\newcommand{\pFq}[5]{\ensuremath{{}_{#1}F_{#2} \left( \genfrac{}{}{0pt}{}{#3}
{#4} \bigg| {#5} \right)}}

\newtheorem{Definition}{\bf Definition}[section]
\newtheorem{Thm}[Definition]{\bf Theorem}
\newtheorem{Example}[Definition]{\bf Example}
\newtheorem{Lem}[Definition]{\bf Lemma}
\newtheorem{Cor}[Definition]{\bf Corollary}
\newtheorem{Prop}[Definition]{\bf Proposition}
\numberwithin{equation}{section}

\section{Introduction}
\label{sec-intro}

Operator algebras provide a mathematical setting upon which many physical 
theories are built. Stellar among these theories is quantum mechanics 
with operator formalism at its core heart. 
The transition from Classical to Quantum Mechanics includes replacing 
the position and momentum by operators 
$p$ and $q$ acting on a function $f$ by 
\begin{equation}
pf = x f(x) \text{ and } qf = \imath \frac{df}{dx},
\end{equation}
\noindent
called the annihilation and creation operators, respectively.  The 
canonical commutation relation of these operators is 
\begin{equation}
\left[ q, p \right] := qp-pq = \imath, \text{ with } \imath = \sqrt{-1}.
\label{noncomm}
\end{equation}
\noindent
Non-commutativity is a common feature in mathematical modeling of reality 
which, in quantum mechanics, introduces the so-called 
Heisenberg-Weyl algebra. This new quality does not come without 
a price $-$ the order of components in operator successions is now relevant 
and has to be carefully
traced in calculations. A traditional solution to this problem is to 
standardize the notation by fixing the order of operators; that is, to use the 
normally ordered expansion in powers of the form $q^kp^j$, in which all 
creation operators stand 
to the left of the annihilation operators. 
 
A \textit{word} in 
the letters $p$'s  and $q$'s is called \textit{balanced} if it contains 
the same number of $p$ and $q$. Theorem \ref{thm-balanced} shows 
that every balanced 
word has a representation as a polynomial in $z = \tfrac{1}{2}(qp+pq)$. 
C. M. Bender and G. V. Dunne \cite{bender-1988a} studied operators in 
symmetrized form $(a_{n,k} = a_{n,n-k}^{*}$, where $*$ denotes complex
conjugation),
\begin{equation}
\sum_{k=0}^{n} a_{n,k} q^{k}p^{n}q^{n-k},
\end{equation}
\noindent
and \cite{bender-1988a} refers the sequence $\{ a_{n,k} \}$ as a 
\textit{pyramid}. In 
this special setup, Theorem \ref{thm-balanced} associates the polynomial 
\begin{equation}
P_{n}(z) = \sum_{k=0}^{n} b_{n,k} z^{k} 
\end{equation}
\noindent
to the sequence $\{ a_{n,k} \}$. The 
relation between $\{ a_{n,k} \}$ and $\{ b_{n,k} \}$ is explicitly given in 
Theorems \ref{formula-an} and \ref{formula-pn} by 
\begin{equation}
a_{n,k} = \frac{1}{\imath^{n} n!}
\sum_{j=0}^{n-k} (-1)^{n-k-j} \binom{n+1}{n-k-j}
\sum_{r=0}^{n} b_{n,r} \imath^{r} \left( j + \tfrac{1}{2} \right)^{r} 
\text{ for } 0 \leq k \leq n,
\label{formula-2a}
\end{equation}
\noindent
and 
\begin{equation}
b_{n,k} = (-1)^{k} \imath^{n+k} 
\sum_{\ell=0}^{n} a_{n,\ell} \sum_{j=0}^{n-k} s(n,j+k) \binom{j+k}{k} 
\left( \ell - \tfrac{1}{2} \right)^{j} \text{ for } 0 \leq k \leq n.
\label{formula-1a}
\end{equation}
\noindent
Here $s(n,k)$ is the Stirling number of the first kind.  \\

Some results pertinent to these two sequences include: 

\begin{proposition}
\label{prop-monic1}
The polynomial $P_{n}(z)$ is monic if and only if $\{ a_{n,k} \}$ is 
normalized by $a_{n,0} + a_{n,1} + \cdots + a_{n,n} = 1$.
\end{proposition}

\begin{theorem}
\label{poly-real1}
The coefficients $\{ b_{n,r} \}$ of $P_{n}(z)$ 
are real if and only if the coefficients $\{ a_{n,k} \}$ are 
hermitian-symmetric; that is if $a_{n,k} = a_{n,n-k}^{*}$. 
\end{theorem}

Parity of the polynomials $P_{n}(z)$ appears from symmetries in the 
pyramid $\{ a_{n,k} \}$.
 
\begin{proposition}
\label{parity-consistent1}
Assume the coefficients $\{ a_{n,k} \}$ are real and symmetric. Then
$P_{n}(z)$ has the same parity as $n$.
\end{proposition}

The next question remains open. 

\begin{problem}
Determine conditions on a real symmetric pyramid $\{ a_{n,k} \}$ in 
order to obtain 
polynomials $\{ P_{n}(z) \}$ which are orthogonal with 
respect to a positive weight function $w(x)$. 
\end{problem}

It is simple to produce algebraic equations for the first 
few coefficients of a 
family of polynomials 
$\begin{displaystyle} C_{n}(z) = \sum_{k=0}^{n} c_{n,k}z^{k} \end{displaystyle}$
in order to be orthogonal. 

\begin{lemma}
\label{ortho-1a}
Assume $ \left\{ C_{n}(z) \right\}$
is a family of monic orthogonal polynomials, with $C_{n}$ of the same 
parity as $n$. Then 
$c_{4,0} + c_{2,0}c_{3,1} - c_{2,0}c_{4,2} = 0$ and 
\begin{equation*}
c_{2,0}c_{5,1} + c_{4,0}c_{5,3} - c_{2,0}c_{4,2}c_{5,3} + c_{6,0} 
-c_{2,0}c_{6,2} - c_{4,0}c_{6,4} + c_{2,0}c_{4,2}c_{6,4} = 0.
\end{equation*}
\end{lemma}

As shown in \cite{bender-1988a}, the 
first condition may be used to prove that certain classical pyramids, such 
as the \textit{symmetric ordering} $a_{n,k} = \delta_{n,k}$ and the 
\textit{Born-Jordan ordering} $a_{n,k} = 1$ do not produce orthogonal 
polynomials. On the other hand, the \textit{Weyl-ordering} $a_{n,k} = 
\binom{n}{k}$ and the case $a_{n,k} = \binom{n}{k}^{2}$ satisfy the 
conditions of Lemma \ref{ortho-1a}.  The polynomials coming from the 
Weyl-ordering may be expressed in terms of the 
\textit{continuous Hahn polynomials} (see Example \ref{weyl-ortho}) and 
those obtained from $a_{n,k} = \binom{n}{k}^{2}$
can be expressed in terms of the 
\textit{Bateman polynomials} (see Example \ref{bateman}).  It 
is curious that these seem to be the only powers of binomial coefficients
that give orthogonal polynomials. 
Experimental evidence on the basis 
of the first condition of Lemma \ref{ortho-1a} rules out the first $50000$ 
power functions $\binom{n}{k}^{r}$.

Partial results on the pyramids $\{ a_{n,k} \}$ 
associated to the Legendre and Hermite polynomials, as examples of 
orthogonal families, and also for the sequence 
$P_{n}(x) = x^{n}$, yield pyramids with combinatorial 
flavor. The complete 
characterization of these pyramids, as well as those corresponding to 
other classical orthogonal polynomials, remains
an open question.

\section{Balanced words in $p$ and $q$ are polynomials in $z$}
\label{sec-balanced}

Let $\mathcal{X} = \{ p, \, q \}$ be an alphabet. A \textit{word} over 
$\mathcal{X}$ is an expression of the form 
\begin{equation}
w = w_{1}w_{2} \cdots w_{k},
\label{gen-w}
\end{equation}
\noindent
with $w_{j} \in \mathcal{X}$. The set of all words is denoted by 
$W(\mathcal{X})$. The 
multiplication of words is defined by concatenation. Every word $w$ over 
$\mathcal{X}$, with $w_{1} = p$, has a unique representation in the form 
\begin{equation}
w = p^{n_{1}} q^{m_{1}} p^{n_{2}}q^{m_{2}}\cdots p^{n_{j}}q^{m_{j}}
\end{equation}
\noindent
with $n_{i}, \, m_{i} \in \mathbb{N}$ with the possibility that $m_{j} = 0$
if the last letter in $w$ is $p$.  A similar unique representation exists
if $w_{1} = q$. 

\begin{definition}
A word $w \in \mathcal{X}$ is called \textit{balanced} if
it has the same number of $p$'s and $q$'s, that is, 
$n_{1}+n_{2}+ \cdots + n_{j} = m_{1}+m_{2}+ \cdots + m_{j}$.
\end{definition}

\begin{definition}
The \textit{free algebra} $F_{\mathbb{C}}(\mathcal{X})$ is the set 
\begin{equation}
F_{\mathbb{C}}(\mathcal{X}) = \left\{ \sum_{j=1}^{m} \alpha_{j}w^{(j)}: 
\, m \in \mathbb{N}, \, \alpha_{j} \in \mathbb{C} \text{ and }
w^{(j)} \text{ is a word over } \mathcal{X} \right\}
\end{equation}
\end{definition}

The results presented here are related to the algebra $\mathcal{A}$ obtained 
from $F_{\mathbb{C}}(\mathcal{X})$ after  the identification $qp-pq = \imath$. 
The fact that this is a non-homogenous element leads to difficulties in 
defining a degree. For instance, the elements $qp^{2}q,\, 2 \imath pq + 
p^{2}q^{2},$ and $\tfrac{1}{4}(qp+pq)^{2} + \tfrac{1}{4}$ all represent 
the same element in $\mathcal{A}$.

\begin{definition}
The \textit{Heisenberg-Weyl algebra} $\mathcal{A}$ is the quotient algebra
\begin{equation}
\mathcal{A} = F_{\mathbb{C}}(\mathcal{X})/ \left\{ qp-pq = \imath \right\}.
\end{equation}
\noindent
Every element $w \in \mathcal{A}$ has a representation in the form 
$\begin{displaystyle} w = \sum_{j=1}^{m} \alpha_{j}w^{(j)}
\end{displaystyle},$ where  
$m \in \mathbb{N}, \, \alpha_{j} \in \mathbb{C}, \,
w^{(j)}$ is a word in $\{p, \, q \}$  and  $qp-pq = \imath$. In 
general, there are many such representations.
\end{definition}

\begin{definition}
Define the special subsets of $\mathcal{A}$ by 
\begin{equation*}
\mathcal{B}[p,q] = \left\{ \sum_{k=0}^{n} \alpha_{n,k}w^{(j)}: \,
\text{ for some } n \in \mathbb{N}, \, \alpha_{n,k} \in \mathbb{C} 
\text{ and } w^{(j)} \text{ balanced word in }p, \,  q
\right\}
\end{equation*}
\noindent
and 
\begin{equation*}
\mathcal{H}[p,q] = \left\{ \sum_{k=0}^{n} a_{n,k}q^{k}p^{n}q^{n-k}: \,
\text{ for some } n \in \mathbb{N},  a_{n,k} \in \mathbb{R} \text { and }
a_{n,k} = a_{n,n-k} \right\}.
\end{equation*}
\noindent
Observe that $\mathcal{H}[p,q] \subset \mathcal{B}[p,q]$, elements in the 
latter are hermitian. 
\end{definition}

The main result of this section is stated next.

\begin{theorem}
\label{thm-balanced}
Every element of 
$\mathcal{B}[p,q]$ is a polynomial in $z = \tfrac{1}{2}(qp+pq)$.
\end{theorem}
\begin{proof}
It suffices to prove the result for a balanced word 
\begin{equation}
w = p^{a_{1}}q^{b_{1}}p^{a_{2}}q^{a_{2}} \cdots p^{a_{t}}q^{b_{t}},
\end{equation}
\noindent
where $a_{j}, \, b_{j} \geq 1$ for $2 \leq j \leq t-1$ and
$a_{1}, \, b_{t} \geq 0$. An easy induction argument proves the 
recurrences
\begin{equation}
p^{c}q = qp^{c} - \left[ (c-1)\imath + 1 \right]p^{c-1} \text{ and }
q^{c}p = pq^{c} + \left[ (c-1)\imath + 1 \right]q^{c-1}.
\end{equation}
\noindent
Without loss of generality, assume that $a_{1} \geq 1$. Then 
\begin{eqnarray*}
w & = & p \left( p^{a_{1}-1}q \right) q^{b_{1}-1}p^{a_{2}}q^{b_{2}} 
\cdots p^{a_{t}}q^{b_{t}} \\
& = & p \left[ qp^{a_{1}-1} - \left( (a_{1}-2) \imath +1 \right)
p^{a_{1}-2} \right] q^{b_{1}-1} p^{a_{2}} \cdots q^{b_{t}} \\
& = & pq \left[ 
p^{a_{1}-1} q^{b_{1}-1} p^{a_{2}} \cdots q^{b_{t}} \right] 
 -  \left( (a_{1}-2) \imath +1 \right)
p^{a_{1}-1} q^{b_{1}-1} p^{a_{2}} \cdots q^{b_{t}}.
\end{eqnarray*}
\noindent
The result now follows by induction on the number of $p$'s and $q$'s in 
the word starting with $qp = z + \imath/2$ and $pq = z - \imath/2$. 
\end{proof}

Theorem \ref{poly-real} and Proposition \ref{parity-consistent} show that 
the polynomials associated to elements of $\mathcal{H}[p,q]$ have real 
coefficients and have the same parity as $n$. 

\smallskip

The next corollary is called a \textit{useful identity} in \cite{bender-1988a}.

\begin{corollary}
The identity $q^{k}p^{n}q^{n-k} = p^{n-k}q^{n}p^{k}$ holds.
\end{corollary}
\begin{proof}
Simply write 
\begin{equation}
q^{k}p^{n}q^{n-k} = \left( q^{k} p^{k} \right) \left( p^{n-k} q^{n-k} 
\right).
\end{equation}
\noindent
The words $q^{k}p^{k}$ and $p^{n-k}q^{n-k}$ are balanced, so they commute. This
gives the result. 
\end{proof}

\section{An expression for a polynomial in two different bases}
\label{sec-two}

Let $n \in \mathbb{N}$ and $a_{n,k} \in \mathbb{C}$. Theorem 
\ref{thm-balanced} gives a polynomial 
map $\mathcal{O}: \mathcal{H}[p,q] \to \mathbb{C}[x]$:
\begin{equation}
\sum_{k=0}^{n} a_{n,k} q^{k}p^{n}q^{n-k} \mapsto 
P_{n}(z):= \sum_{r=0}^{n} b_{n,r}z^{r}.
\label{relation-1}
\end{equation}
\noindent
Explicit formulas connecting 
 $\{a_{n,k} \}$ and $\{ b_{n,k} \}$ are given in this section.

\begin{theorem}
\label{formula-an}
The sequence $\{ a_{n,k} \}$ is given by 
\begin{equation}
a_{n,k} = \frac{1}{\imath^{n} n!}
\sum_{j=0}^{n-k} (-1)^{n-k-j} \binom{n+1}{n-k-j}
P_{n} \left( \imath \left( j + \tfrac{1}{2} \right) \right) 
\text{ for } 0 \leq k \leq n.
\label{form-an}
\end{equation}
\noindent
Expanding $P_{n}$ gives \eqref{formula-2a}.
\end{theorem}
\begin{proof}
The realization
$p = x$  and $q = \imath  \frac{d}{dx}$ gives 
$\tfrac{1}{2}(qp + pq) \left(x^{m} \right) = 
\imath  \left( m + \tfrac{1}{2} \right) x^{m}$ \text{ and }
$q^{k}p^{n}q^{n-k}\left(x^{m} \right) 
= \imath^{n} n! \binom{m+k}{m+k-n}x^{m}$,
with the usual convention that $\binom{a}{b} = 0$ if $b<0$. It follows that 
\begin{equation}
\sum_{k=0}^{n} a_{n,k} q^{k}p^{n}q^{n-k} \left(x^{m} \right)
 = \imath^{n}n! \sum_{\ell = 0}^{m} \binom{n+\ell}{\ell}
a_{n,n-m+\ell} x^{m}
\end{equation}
\noindent
and 
\begin{equation}
P_{n} \left( \tfrac{1}{2}(qp+pq)  \right)x^{m} = 
P_{n} \left( \imath \left( m + \tfrac{1}{2} 
\right) \right) x^{m}.
\end{equation}
\noindent
Therefore 
\begin{equation}
P_{n} \left( \imath \left(m + \tfrac{1}{2} \right) \right) = 
\imath^{n} n! \sum_{\ell=0}^{m} \binom{n + \ell }{\ell} 
a_{n,n-m+\ell}, \quad \text{ for } 0 \leq m \leq n.
\label{formula-1}
\end{equation}
Then \eqref{form-an} 
is obtained by solving the linear system \eqref{formula-1}
for $a_{n,m}$, and using the formula for
matrix $M_{n}^{-1}$ given in the next statement.
\end{proof}

\begin{lemma}
The inverse of the Hankel matrix
$\begin{displaystyle} M_{n} = \left[ \binom{i+j}{n} \right]_{0 \leq i,j 
\leq n} \end{displaystyle}$ is 
\newline
$\begin{displaystyle} M_{n}^{-1} = \left[ 
(-1)^{n-i-j} \binom{n+1}{i+j+1} \right]_{0 \leq i,j 
\leq n} \end{displaystyle}$.
\end{lemma}
\begin{proof}
The claim is equivalent to the identity
\begin{equation}
\sum_{k=0}^{n} (-1)^{n-i-j} \binom{i+k}{n} \binom{n+1}{k+j+1} 
= 
\begin{cases}
0 & \quad \text{ if } i \neq j, \\
1 & \quad \text{ if } i =  j.
\end{cases}
\label{teddy-11}
\end{equation}
\noindent
The proof of \eqref{teddy-11} is a routine application of WZ.  \\
 
\noindent
\textbf{An alternative proof}. Consider the Jordan 
block matrix
$J_{n+1} = J_{n+1}(i,j)$ with zero entries except for $1$ when $j = i+1$. Then 
$J_{n+1}^{k}(i,j)$ is zero except a shifted diagonal with $1's$ at $j=i+k$. In 
particular, $J_{n+1}^{k}=0\,\,\text{for}\,\, k>n$. Thus 
$\tilde{M}_{n}$ can be expressed as
$\begin{displaystyle} 
\tilde{M}_{n}  =  \sum_{k=0}^{n}\binom{n+k}{k}J_{n+1}^{k}
=\sum_{k=0}^{+\infty}\binom{n+k}{k}J_{n+1}^{k}
\end{displaystyle}$. Now use 
$\begin{displaystyle} 
\sum_{k=0}^{+\infty}\binom{n+k}{k}z^{k}=\left(1-z\right)^{-n-1}
\end{displaystyle}$ to conclude that 
$\begin{displaystyle}\tilde{M}_{n}=\left(I_{n+1}-J_{n+1}\right)^{-n-1}
\end{displaystyle}$. Thus 
$$\begin{displaystyle}
\tilde{M}_{n}^{-1}  =  \left(I_{n+1}-J_{n+1}\right)^{n+1}
 =  \sum_{k=0}^{n}\binom{n+1}{k}\left(-1\right)^{k}J_{n+1}^{k}
\end{displaystyle}$$
\noindent
which proves the result.
\end{proof}

The expression for $a_{n,k}$ is particularly simple in the outer diagonal
$\{ a_{n,n} \}$. 

\begin{proposition}
\label{outer-b1}
Let $\{ a_{n,k} \}$ be a pyramid with corresponding polynomials $\{P_{n} \}$.
The outer diagonal of the pyramid is given by
\begin{equation}
a_{n,n} = \frac{1}{\imath^{n} n!} P_{n} \left( \frac{\imath}{2} \right).
\label{outer-an}
\end{equation}
\noindent
Therefore, if the polynomials $P_{n}$ have an exponential
generating function
\begin{equation}
G(z,t)=\sum_{n=0}^{\infty} \frac{P_{n}(z)}{n!}t^{n}
\end{equation}
then the horizontal generating function for the outer diagonal is 
\begin{equation}
\sum_{n=0}^{\infty} a_{n,n}t^{n}=
G\left(\frac{\imath}{2},\frac{t}{\imath}\right).
\end{equation}
\end{proposition}

The authors of \cite{bender-1988a} state that \textit{apparently, the classical 
orthogonal polynomials give pyramids with ugly entries}. The result of 
Proposition \ref{outer-b1} gives expressions for the outer diagonal 
$\{ a_{n,n} \}$. 

\begin{example}
The Legendre polynomials 
\begin{equation}
P_{n}(x) = \frac{1}{\binom{2n}{n}} 
\sum_{m \geq 0} (-1)^{m} \binom{n}{m} \binom{2n-2m}{n} x^{n-2m},
\end{equation}
\noindent
normalized to be monic, form a sequence of orthogonal
polynomials. The corresponding (non-normalized) pyramid is
\begin{equation}
\begin{array}{ccccccccccccccc}
& &   &  &    &     &    &  1  &    &      &    &   &   &   \nonumber \\
& &   &  &    &     &  1 &  & 1   &   &  & & & \nonumber \\
& &   &  &    &  7   &    & 10   &   &  7
   &    &   &   &   \nonumber \\
& &   &  & 17  &    &  103  & 
  & 103  &    &   17 &   &   &   \nonumber \\
& &   & 203  &   &  2948 &  &  
7138  &  &  2948   &   &  203 & 
  &   \nonumber \\
& & 583 & &  20091  &  &  
 100286 &  &  100286 &  &  
20091 & & 583 &   \nonumber \\
\end{array}
\end{equation}

Proposition \ref{outer-b1} gives
\begin{equation}
a_{n,n} = \frac{n!}{2^{n} (2n)!} 
\sum_{j=0}^{n} 2^{2j} \binom{n}{j} \binom{2n-2j}{n}.
\end{equation}
\end{example}

\begin{example}
The (monic) Hermite polynomials are defined by 
\begin{equation}
H_{n}(x) = \frac{n!}{2^{n}} \sum_{m=0}^{\lfloor{ \tfrac{n}{2} \rfloor}}
 \frac{(-1)^{m}}{m!(n-2m)!} 
(2x)^{n-2m}.
\end{equation}
\noindent
The corresponding pyramid is
\begin{equation}
\begin{array}{ccccccccccccccc}
& &   &  &    &     &    &  1  &    &      &    &   &   &   \nonumber \\
& &   &  &    &     &  1 &    & 1   &      &    &   &   &   \nonumber \\
& &   &  &    & 3   &    & 2   &   &  
3    &    &   &   &   \nonumber \\
& &   &  & 7   &    &  17  &   & 
17  &    &   7 &   &   &   \nonumber \\
& &   & 25  &   &  76  &  &  182
 &  &  76 &   &  25 &   &   \nonumber \\
& & 27  & &  159 &  &  
454 &  & 454  &  &  159
 & & 27  &   \nonumber \\
\end{array}
\end{equation}

The information to state the next Lemma came from OEIS, entry $A047974$. 

\begin{lemma}
The outer diagonal sequence $\{ a_{n,n} \}$ of the pyramid corresponding to
the Hermite polynomials is given by 
$a_{n,n} = h_{n}/2^{n}n!$,
where $\{ h_{n} \}$ satisfies the recurrence 
$h_{n} = h_{n-1} + 2(n-1)h_{n-2}, \quad h_{1} = 1, \, h_{2} = 3$. An
explicit representation of $\{ h_{n} \}$ is given by 
\begin{equation}
h_{n}  =  \frac{1}{2 \sqrt{\pi}} \int_{-\infty}^{\infty} x^{n} 
e^{-(x-1)^{2}/4} \, dx, 
\end{equation}
\noindent
so that $h_{n}$ can be interpreted as the moment of order $n$ of 
a Gaussian random variable $X$, with mean $1$ and variance $2$.
\end{lemma}
\begin{proof}
The sequence $\{ a_{n,n} \}$ is related to the (monic) Hermite polynomials by
the statement of Proposition \ref{outer-b1}. The details follow from the 
recurrence for the Hermite polynomials:
$\begin{displaystyle} 2H_{n+1}(x) = 2xH_{n}(x) - nH_{n-1}(x).
\end{displaystyle}$
\end{proof}
\end{example}

\begin{example}
The Chebyshev polynomials $T_{n}$ are given by 
\begin{equation}
T_{n}(x) = \frac{1}{2} \left[ ( x + \sqrt{x^{2}-1} )^{n} + 
( x - \sqrt{x^{2}-1} )^{n} \right].
\end{equation}
\noindent
In this case \eqref{outer-an} gives 
\begin{equation*}
a_{n,n}  =  \frac{1}{\imath^{n} n!} T_{n} \left( \frac{\imath}{2} 
\right) 
 =  \frac{1}{2n!} \left[ \left( \frac{1+ \sqrt{5}}{2} \right)^{n} +
\left( \frac{1- \sqrt{5}}{2} \right)^{n} \right].
\end{equation*}
\noindent
This yields $a_{n,n} = L_{n}/2n!$ where 
$L_{n}$ is the \textit{Lucas number}. The 
corresponding result for the Chebyshev polynomials of the second 
kind $U_{n}(x)$ gives 
$a_{n,n} = F_{n+1}/n!$, with 
$F_{n}$ is the \textit{Fibonacci number}. The proof is based on 
Binet's formula 
\begin{equation}
F_{n} = \frac{ 
\left( 1+ \sqrt{5} \right)^{n} - 
\left( 1- \sqrt{5} \right)^{n} 
}{2^{n} \, \sqrt{5}}.
\end{equation}
\end{example}

Theorem \ref{formula-an} provided an expression for $\{ a_{n,k} \}$ in 
terms of the associated polynomial sequence $\{ P_{n} \}$. The 
next result gives the polynomial $P_{n}(x)$ in terms of the 
sequence $\{ a_{n,k} \}$.

\begin{theorem}
\label{formula-pn}
The polynomials $P_{n}(z)$ associated to the sequence 
$\{ a_{n,k} \}$ are given by 
\begin{equation}
P_{n}(z) = \imath^{n} n! \sum_{k=0}^{n} a_{n,k} 
\binom{- \imath z - \tfrac{1}{2} +k}{n}
= \imath^{n} \sum_{k=0}^{n} a_{n,n-k} 
\left(- \imath z + \tfrac{1}{2} -k \right)_{n},
\label{poly-nice2}
\end{equation}
where $(x)_{n}$ is the shifted factorial defined by 
$(x)_{n} = x(x+1)\cdots (x+n-1)$. Expanding in powers of $z$ 
gives \eqref{formula-1a}.
\end{theorem}
\begin{proof}
The result follows directly from formula 
\eqref{formula-1}, that gives  $P_{n}(\imath(m+ \tfrac{1}{2}))$ for 
$0 \leq m \leq n$ (a total of $n+1$ points), remarking that 
\eqref{formula-1} holds in fact for all $m \in \mathbb{R}$ when written in the 
equivalent form 
\begin{equation}
P_{n} \left( \imath \left( m + \tfrac{1}{2} \right) \right) = 
\imath^{n} n! \sum_{k=0}^{n} a_{n,k} \binom{m+k}{n}.
\end{equation}
\noindent
To obtain \eqref{formula-1a}, use 
$\begin{displaystyle}(y)_{k} = \sum_{j=0}^{k} s(k,j) (y+k-1)^{j}
\end{displaystyle}$, the generating function of the Stirling numbers.
\end{proof}

The $1$-dimensional \it Weyl algebra \rm $\mathcal{A}_1$ is the free algebra 
with two generators $R$ and $D$ together with
the commutation relation $RD-DR=1$. This is a parallel version of the 
Heisenberg-Weyl 
algebra $\mathcal{A}$ discussed in Section \ref{sec-balanced}. Each 
$u\in\mathcal{A}_1$ can be 
expressed uniquely in the normal form
$u=\sum_{j,k}\alpha_{j,k}R^jD^k$. 
The connection to lattice paths or Ferrer diagrams is natural and well-known. 
To see this, assume 
$u$ has $n$ and $m$ letters of $R$ and $D$,
respectively. Construct a walk from $(0,m)$ to $(n,0)$ as follows: by 
reading $u$ from left to right, move a unit right
(resp. down) step if the letter is $R$ (resp. $D$). See references 
\cite{blasiak-2008a,graham-1994a,navon-1973a,varvak-2005a} for details. 

The algebra $\mathcal{A}_{1}$ allows alternative proofs for some 
of the results given in this section. 
Among the properties of $\mathcal{A}_{1}$ used, the following ones are easy to 
establish by induction:
\begin{equation}
\label{two-facts}
D^kR^k =\prod_{j=1}^k(DR-j+1), \qquad
R^kD^n =\sum_{j=0}^k\binom{k}j(n-j+1)_jD^{n-j}R^{k-j}.
\end{equation}
If $D= \imath q$ and $R=p$, then 
\begin{equation}
RD-DR=1, \quad 
\mathcal{O}(q^np^n)=\imath^n \mathcal{O}(D^nR^n).
\end{equation}
In this setting, the proof 
of Theorem \ref{formula-pn} begins with 
the introduction of $x:=RD+DR=2DR+1$, so that $DR=\frac{x-1}2$. Note that 
$D = iq$ and $R=p$ yields $x = \imath(pq+qp) = 2 \imath z$. To 
complete the proof, use \eqref{two-facts} to find
\begin{eqnarray*}
\mathcal{O}(q^np^n) &= & \imath^n\sum_{k=0}^{n} a_{n,k}\sum_{j=0}^{k}
\binom{k}j(n-j+1)_jD^{n-j}R^{k-j}R^{n-k}\\
&= & \imath^n\sum_{k=0}^{n}a_{n,k}\sum_{j=0}^{k}\binom{k}j(n-j+1)_jD^{n-j}R^{n-j} \\
&= & \imath^n\sum_{k=0}^{n}a_{n,k}\sum_{j=0}^{k}\binom{k}j(n-j+1)_j
\prod_{m=1}^{k}\left(\tfrac12x-m+\tfrac12\right).
\end{eqnarray*}
The transformation of the last expression to \eqref{poly-nice2}
is automatic with the WZ method.

\begin{definition}
The polynomials appearing in Theorem \ref{formula-pn} are denoted by 
\begin{equation}
Q_{n,k}(z) := \binom{ - \imath z - \tfrac{1}{2} +k}{n} = 
\frac{1}{n!} \prod_{\ell=0}^{n-1} ( - \imath z - \tfrac{1}{2} + k - \ell ).
\label{formula-qn}
\end{equation}
\noindent
These are the  polynomials associated to the homogeneous elementary words 
considered by Bender-Dunne:
\begin{equation}
Q_{n,k}(z) = \frac{1}{\imath^{n} n!} 
\mathcal{O} \left( q^{k}p^{n} q^{n-k} \right).
\label{QmapO}
\end{equation}
\end{definition}

\begin{theorem}
\label{q-li}
The polynomials $\{ Q_{n,k}(z): \, 0 \leq k \leq n \}$ form a basis for the 
vector space of polynomials of degree at most $n$.
\end{theorem}
\begin{proof}
Each $Q_{n,k}(z)$ is a polynomial of degree $n$, so it 
suffices to establish their linear independence. Fix $n \in \mathbb{N}$. 
From \eqref{formula-1a} and \eqref{QmapO} it follows that the 
coefficient of $x^{r}$ in $Q_{n,k}(x)$ is 
\begin{equation}
\frac{(-\imath)^{r}}{n!} \sum_{j=0}^{n-r} s(n,j+r) \binom{j+r}{r} 
\left( k- \tfrac{1}{2} \right)^{j}.
\end{equation}
\noindent
Now suppose 
$\begin{displaystyle} \sum_{k=0}^{n} u_{k}(n) Q_{n,k}(x) = 0. 
\end{displaystyle}$ The 
vanishing of the coefficient of $x^{r}$ gives the system 
of equations 
\begin{equation}
\sum_{k=0}^{n} u_{k}(n) \sum_{j=0}^{n-r} s(n,j+r) \binom{j+r}{r} 
\left( k - \tfrac{1}{2} \right)^{j} = 0 
\label{ld1}
\end{equation}
\noindent
for $0 \leq r, k \leq n$. Let $S = (S_{r,k})$ be the $(n+1) \times (n+1)$ 
matrix with entries 
\begin{equation}
S_{r,k} =  \sum_{j=0}^{n-r} s(n,j+r) \binom{j+r}{r} 
\left( k - \tfrac{1}{2}  \right)^{j},
\end{equation}
\noindent
so that \eqref{ld1} is $Su = 0$, where $u$ is the vector $(u_{k}(n))$. The 
independence of $\{ Q_{n,k}(x) \}$ is 
equivalent to the invertibility of $S$. Observe the factorzation
$S = XY$ where  $X = (X_{a,b})$ and $Y = (Y_{a,b})$ with 
$X_{a,b} =  \binom{a+b}{b} s(n,a+b)$  and  
$Y_{a,b} = \left( b - \tfrac{1}{2} \right)^{a}$. The 
matrix $X$ is upper triangular, $x_{a,b} = 0$ if $a+b > n$ and $Y$ 
is a Vandermonde matrix. Therefore 
\begin{equation}
\det X = \prod_{a=0}^{n} \binom{n}{a} \text{ and }
\det Y = \prod_{j < k } \left[ \left( k - \tfrac{1}{2} \right) - 
\left( j - \tfrac{1}{2} \right) \right] =
\prod_{i=0}^{n} i!,
\end{equation}
\noindent
proving that $\begin{displaystyle} \det S = \prod_{k=1}^{n} k^{k} 
\end{displaystyle}$ and $S = XY$ is invertible.
\end{proof}

\begin{corollary}
\label{thm-basis}
The balanced words 
$\left\{ q^{k}p^{n}q^{n-k}: \, 0 \leq k \leq n \right\}$ form 
a basis for the class of balanced words of weight $n$. 
\end{corollary}

The next result phrases Theorem \ref{q-li} in the language of the Weyl 
algebra $\mathcal{A}_{1}$.

\begin{theorem}
The set 
$ S = \left\{ R^{j} D^{k}: \, j, \, k \geq 0 \right\} $ 
is a basis for $\mathcal{A}_{1}$.
\end{theorem}
\begin{proof}
As before, it suffices to verify linear independence. Let 
$K = \mathbb{C}[x,y]$ be a commutative polynomial ring, with basis 
$\{ x^{j}y^{k}: \, j, \, k \geq 0 \}$ over $\mathbb{C}$. Define the linear 
operators $R$ and $D$ on $K$ by
\begin{equation}
R \left( x^{j} y^{k} \right) = x^{j+1}y^{k} \text{ and }
D \left( x^{j} y^{k} \right) = jx^{j-1}y^{k} + x^{j}y^{k+1}.
\end{equation}
\noindent
A direct calculation shows
that $RD-DR = 1$ on $K$, and hence $\mathbb{C}[x,y]$ is a 
representation of the Weyl
algebra $\mathcal{A}_{1}$. To verify linear independence, suppose 
\begin{equation}
L:= \sum_{j,k} \alpha_{j,k} R^{j}D^{k} = 0. 
\end{equation}
\noindent
A direct computation gives $D^{k}(1) = y^{k}$ and $R^{j}(y^{k}) = 
x^{j}y^{k}$. Thus, the 
value $L(1) = 0$ gives $\alpha_{j,k} = 0$ which proves independence.
\end{proof}

\smallskip

\noindent
\textbf{Problem}. A word $w = w_{0}w_{1} \cdots w_{n}$ in $\mathcal{A}_{1}$ 
(where each $w_{k}=p$ or $q$) is a \textit{palindrome} if 
$w_{n-k}= w_{k}$ for $0 \leq k \leq n$. The `adjoint' of a monomial word 
$w= w_{0}w_{1} \cdots w_{n}$ is the word $w^{*} = w_{n} \cdots w_{1}w_{0}$. 
This is extended to $z$ in $\mathcal{A}_{1}$ by linearity. 
The element $z$ is called Hermitian if $z = z^{*}$. For instance $p$ and 
$q$ are Hermitians, but $pq$ is not, since $(pq)^{*} = q^{*}p^{*} = qp$. 
\noindent
\textit{Question}: Is it true that a monomial word $w$ is Hermitian if and 
only if $w$ is a palindrome? It is clear that if $w$ is a palindrome, then 
$w$ is Hermitian. The questions is to decide on the converse. 

\section{Polynomials versus pyramids}
\label{sec-relations}

This section discusses how certain properties of the pyramids $\{ a_{n,k} \}$ 
are reflected on the corresponding polynomials $P_{n}(z)$.

\begin{proposition}
\label{prop-monic}
The polynomial $P_{n}(z)$ is monic if and only if $\{ a_{n,k} \}$ is 
normalized by $a_{n,0} + a_{n,1} + \cdots + a_{n,n} = 1$.
\end{proposition}
\begin{proof}
The polynomial 
\begin{equation}
\binom{-\imath z - \tfrac{1}{2} + k}{n} = 
\frac{1}{n!} \prod_{j=0}^{n-1} (-\imath z - \tfrac{1}{2} +k-j)
\end{equation}
\noindent
has leading coefficient $(- \imath )^{n}/n!$. Theorem \ref{formula-pn}
now shows that the leading 
coefficient of $P_{n}(z)$ is the sum of $\{ a_{n,k} \}$.
\end{proof}

The next statement clarifies the condition of hermitian-symmetry 
imposed on the pyramids \cite{bender-1988a}. The 
analysis begins the following observation.

\begin{lemma}
\label{q-conj}
The polynomial $Q_{n,k}(x)$ defined in \eqref{formula-qn}, with 
$x \in \mathbb{R}$,  satisfies  the symmetry identity
$Q_{n,k}^{*}(x) = (-1)^{n} Q_{n,n-k}(x).$
\end{lemma}
\begin{proof}
This follows directly from 
\begin{equation*}
Q_{n,k}^{*}(x)  =  \frac{1}{n!} \prod_{\ell=0}^{n-1} 
\left(- \imath x - \tfrac{1}{2} + k - \ell \right)^{*} 
=  \frac{(-1)^{n}}{n!} \prod_{\ell=0}^{n-1} 
\left(- \imath x + \tfrac{1}{2} - k + \ell \right).
\end{equation*}
\end{proof}

The next result characterizes real polynomials $P_{n}$. 

\begin{theorem}
\label{poly-real}
The coefficients $\{ b_{n,r} \}$ of $P_{n}(z)$ 
are real if and only if the coefficients $\{ a_{n,k} \}$ are 
hermitian-symmetric; that is if $a_{n,k} = a_{n,n-k}^{*}$. 
\end{theorem}
\begin{proof}
The identity 
$\begin{displaystyle} P_{n}(x)^{*}   =  
(- \imath )^{n} n! \sum_{k=0}^{n} a_{n,k}^{*} 
Q_{n,k}(x) \end{displaystyle}$ and Lemma \ref{q-conj} show that 
$b_{n,r}$ real is equivalent to 
\begin{equation}
\sum_{k=0}^{n} a_{n,k}Q_{n,k}(x) = 
\sum_{k=0}^{n} a_{n,n-k}^{*}Q_{n,k}(x).
\end{equation}
\noindent
The result now follows from Theorem \ref{q-li}.
\end{proof}

\begin{proposition}
\label{parity-consistent}
Assume the coefficients $\{ a_{n,k} \}$ are real and symmetric. Then
$P_{n}$ has the same parity as $n$. 
\end{proposition}
\begin{proof}
Use the identity
$Q_{n,k}(-x) = (-1)^{n}Q_{n,n-k}(x)$
established in the same way as in the proof of Lemma \ref{q-conj}.
\end{proof}

\section{Necessary conditions for orthogonality}
\label{sec-neces}

Properties of the pyramid $\{ a_{n,k} \}$ reflect on those of the 
associated sequence of polynomials $\{ P_{n} \}$. For instance, if 
$\{ a_{n,k} \}$ is normalized (total sum equal to $1$), real
and symmetric $(a_{n,k} = a_{n,n-k})$, then $\{ P_{n} \}$
are monic, with real coefficients and $P_{n}$ has the same parity as $n$. The 
question 
considered in this section is to determine conditions on
$\{ a_{n,k} \}$ that yield orthogonal polynomials.  

Recall that a family of 
polynomials $\{ C_{n} \}$ is called orthogonal if there is a positive 
weight function $w(z)$  such that 
\begin{equation}
\langle C_{n}, C_{m} \rangle :=  \int_{\mathbb{R}} C_{n}(z)C_{m}(z) w(z) \, dz 
= 
\begin{cases}
w_{n} > 0 & \text{ if } n = m, \\
 0 & \text{ if } n \neq  m.
\end{cases}
\end{equation}

Now assume that 
$\begin{displaystyle} C_{n}(z) = \sum_{k=0}^{n} c_{n,k}z^{k} 
\end{displaystyle}$  is 
a family of monic, orthogonal polynomials with $C_{n}$ of the same 
parity as $n$. The 
orthogonality of $\{ C_{n} \}$ yields a sequence of algebraic equations 
that the coefficients $\{ c_{n,k} \}$ must satisfy. For instance, 
$\langle C_{0}, C_{2} \rangle = 0$ gives 
\begin{equation}
\left[z^{2} \right] + c_{2,0} \left[ 1 \right] = 0,
\end{equation}
\noindent
where 
\begin{equation}
\left[ z^{a} \right] = \int_{\mathbb{R}} z^{a} w(z) \, dz.
\end{equation}
\noindent
Similarly, the orthogonality of the pairs $\{ C_{0}, C_{4} \}$ and 
$\{ C_{1}, C_{3} \}$ yield 
\begin{eqnarray}
\left[ z^{4} \right] + c_{4,2} \left[ z^{2} \right]  
+ c_{4,0} \left[ 1 \right] & = & 0 \label{eqn-1} \\
\left[ z^{4} \right] + c_{3,1} \left[ z^{2} \right]  & = & 0.
\nonumber
\end{eqnarray}
\noindent
The vanishing of the corresponding determinant gives 
\begin{equation}
c_{4,0} + c_{2,0}c_{3,1} - c_{2,0}c_{4,2} = 0.
\label{condition-1}
\end{equation}
\noindent
Looking at the first six polynomials gives, among many, the relation 
\begin{equation}
c_{2,0}c_{5,1} + c_{4,0}c_{5,3} - c_{2,0}c_{4,2}c_{5,3} + c_{6,0} 
-c_{2,0}c_{6,2} - c_{4,0}c_{6,4} + c_{2,0}c_{4,2}c_{6,4} = 0.
\label{condition-2}
\end{equation}

\begin{example}
The \textit{symmetric ordering} has $a_{n,0} = a_{n,n} = \tfrac{1}{2}$
and $a_{n,k} = 0 $ if $k \neq 0, n$. Theorem \ref{formula-pn} gives 
\begin{equation}
P_{n}(z) = \frac{\imath^{n}}{2} 
\left[ (- \imath z + \tfrac{1}{2} - n )_{n} + 
(- \imath z + \tfrac{1}{2} )_{n}  \right].
\end{equation}
\noindent
The first few values appear in \cite{bender-1988a}: 
\begin{equation*}
P_{0}(z) = 1, \, 
P_{1}(z) = z, \, P_{2}(z) = z^{2} - \tfrac{3}{4}, \,
P_{3}(z) = z^{3} - \tfrac{23}{4}z, \, 
P_{4}(z) = z^{4} - \tfrac{43}{4}z^{2} + \tfrac{105}{16}. 
\end{equation*}
\noindent
This sequence of polynomials is not orthogonal since condition
\eqref{condition-1} is not satisfied.
\end{example}

\begin{example}
The polynomials corresponding to the \textit{Born-Jordan ordering}
have $a_{n,k} = \tfrac{1}{n+1}$ for $0 \leq k \leq n$. Theorem 
\ref{formula-pn} gives 
\begin{equation}
P_{n}(z) = \frac{\imath^{n}}{n+1} \sum_{k=0}^{n}  
\left( - \imath z + \tfrac{1}{2} - k \right)_{n}.
\end{equation}
\noindent
The first few values also appear in \cite{bender-1988a}: 
\begin{equation*}
P_{0}(z) = 1, \, 
P_{1}(z) = z, \,\, P_{2}(z) = z^{2} - \tfrac{5}{12}, \,
P_{3}(z) = z^{3} - \tfrac{11}{4}z, \, 
P_{4}(z) = z^{4} - \tfrac{19}{2}z^{2} + \tfrac{189}{80}. 
\end{equation*}
\noindent
Condition \eqref{condition-1} does not hold, so this sequence 
is not orthogonal.
\end{example}

\begin{example}
\label{weyl-ortho}
The \textit{Weyl ordering}  has
$a_{n,k} = 2^{-n}\binom{n}{k}$, and 
\eqref{poly-nice2} give the polynomials
\begin{equation}
P_{n}(z) = \frac{\imath^{n}}{2^{n}} 
\sum_{k=0}^{n} \binom{n}{k} ( - \imath z + \tfrac{1}{2} - k)_{n}.
\label{weyl-10}
\end{equation}
\noindent
The first few values may be found in \cite{bender-1988a}: 
\begin{equation*}
P_{0}(z) = 1, \, 
P_{1}(z) = z, \,\, P_{2}(z) = z^{2} - \tfrac{1}{4}, \,
P_{3}(z) = z^{3} - \tfrac{5}{4}z, \, 
P_{4}(z) = z^{4} - \tfrac{7}{2}z^{2} + \tfrac{9}{16}. 
\end{equation*}
\noindent
The condition \eqref{condition-1} is now satisfied, so one might expect that 
these polynomials form an
orthogonal family. It is
stated in \cite{bender-1988a} that
\begin{equation}
P_{n}(z) = \frac{n!}{(2 \imath)^{n}} 
\pFq32{-n,n+1, \tfrac{1}{4} - \tfrac{\imath z}{2} }{\tfrac{1}{2}, 1 }{1}.
\label{rep-weyl}
\end{equation}
\noindent
To verify \eqref{rep-weyl} from \eqref{weyl-10} is equivalent to the identity
\begin{equation}
\sum_{k=0}^{n} \binom{n}{k} ( - \imath z + \tfrac{1}{2} -k )_{n} = 
(-1)^{n} n! 
\pFq32{-n,n+1, \tfrac{1}{4} - \tfrac{\imath z}{2} }{\tfrac{1}{2}, 1 }{1}.
\label{desired-1}
\end{equation}

The left-hand side of \eqref{desired-1} is written now in hypergeometric 
form. 

\begin{lemma}
For $n \in \mathbb{N}$
\begin{equation}
\sum_{k=0}^{n} \binom{n}{k} ( - \imath z + \tfrac{1}{2} - k )_{n} = 
( - \imath z - n + \tfrac{1}{2})_{n} \, \, 
\pFq21{-n,\tfrac{1}{2}- \imath z}{\tfrac{1}{2} -n - \imath z}{-1}.
\label{iden-hyper-1a}
\end{equation}
\end{lemma}
\begin{proof}
It suffices to verify 
\begin{equation}
\sum_{k=0}^{n} \binom{n}{k} (m+1-k)_{n} = 
(m-n+1)_{n}  \, \, 
\pFq21{-n,m+1}{m -n +1}{-1}
\end{equation}
\noindent
obtained from \eqref{iden-hyper-1a}
with $z = \imath ( m + \tfrac{1}{2} )$, as both sides are polynomials 
in $z$. This is accomplished by writing 
the left-hand side as 
\begin{equation}
\sum_{k=0}^{n} \frac{n!}{(n-k)! \, k!} \frac{(m+1-n+k)_{n}}{(m-n+1)_{n}}
\end{equation}
\noindent
and using
\begin{equation}
\frac{n!}{(n-k)!} = (-1)^{n} (-n)_{k} \text{ and }
\frac{(m+1-n+k)_{n}}{(m-n+1)_{n}} = \frac{(m+1)_{k}}{(m+1-n)_{k}}.
\end{equation}
\noindent
This proves the result.
\end{proof}

Therefore \eqref{rep-weyl} is equivalent to 
\begin{equation}
(- \imath z - n + \tfrac{1}{2} )_{n} \, 
\pFq21{-n,\tfrac{1}{2} - \imath z}{\tfrac{1}{2} -n - \imath z}{-1} = 
(-1)^{n} n! \, 
\pFq32{-n,n+1,\tfrac{1}{4} - \tfrac{\imath z}{2}}{\tfrac{1}{2},1}{1}.
\label{iden-99a}
\end{equation}
\noindent
This is proved by observing again that both sides are polynomials 
in $z$, so it suffices 
to verify \eqref{iden-99a}
when $z = \imath \left( m + \tfrac{1}{2} \right)$ and $m \in \mathbb{N}$. The 
identity becomes 
\begin{equation}
\label{form-hyp1}
\frac{(m-n+1)_{n}}{n!} \pFq21{-n, m+1}{m-n+1}{-1} = (-1)^{n} 
\pFq32{-n,n+1,\tfrac{m+1}{2}}{\tfrac{1}{2},1}{1}.
\end{equation}
\noindent
In detail, 
\begin{equation}
\frac{(m-n+1)_{n}}{n!} 
\sum_{k=0}^{n} \frac{(-1)^{k}(-n)_{k} (m+1)_{k}}{k! (m-n+1)_{k}} =
(-1)^{n} \sum_{k=0}^{n} 
\frac{(-n)_{k}(n+1)_{k} \left( \tfrac{m+1}{2} \right)_{k}}
{(1)_{k} k! \left( \tfrac{1}{2} \right)_{k}}.
\label{hahn-iden2}
\end{equation}
\noindent
The method of Wilf-Zeilberger (WZ) 
described in \cite{petkovsek-1996a} shows that 
both sides of \eqref{hahn-iden2}, with $m$ fixed, satisfy the recurrence 
\begin{equation}
(n+2)u(n+2,m) - (2m+1)u(n+1,m) - (n+1)u(n,m)=0.
\label{wz-1}
\end{equation}
Upon verifying the value for $n=0$ and $n=1$, the assertion 
\eqref{hahn-iden2} follows. 
\end{example}

\begin{note}
The result \eqref{form-hyp1} 
also follows from entries $102$ and $103$ in page $540$ of 
\cite{prudnikov-1998a}. The stated identity is 
\begin{equation}
\pFq32{-n,a,b}{\tfrac{a-n}{2},\tfrac{1+a-n}{2}}{1} =
\frac{(2b-a+1)_{n}}{(1-a)_{n}} 
\pFq21{-n,2b}{2b-a+1}{-1},
\end{equation}
\noindent
and \eqref{form-hyp1} is obtained by choosing 
$a=n+1, \, b = \tfrac{m+1}{2}$. 
\end{note}

\begin{note}
The left-hand side of \eqref{hahn-iden2} can be reduced to 
\begin{equation}
c(n,m) := \sum_{k=0}^{n} \binom{m+k}{m} \binom{m}{n-k} = 
\sum_{k=0}^{m} \binom{n+k}{m} \binom{m}{k}.
\end{equation}
\noindent 
The equality of the two versions of $c(n,m)$ follows from \eqref{wz-1}. These 
coefficients have remarkable properties which is a topic deferred to
\cite{amdeberhan-2013b}. 
\end{note}

\begin{note}
The hypergeometric representation of the polynomials $P_{n}(z)$ given in 
\eqref{rep-weyl} shows that $P_{n}(z)$ may be expressed in terms of the 
\textit{continuous Hahn polynomials}
\begin{equation*}
p_{n}(z;a,b,c,d) = \imath^{n} \frac{(a+c)_{n} (a+d)_{n}}{n!} 
\pFq32{-n,n+a+b+c+d-1,a+ \imath z}{a+c,a+d}{1}.
\end{equation*}
\noindent
The identity \eqref{rep-weyl} shows that
\begin{equation}
P_{n}(z) = \frac{(-1)^{n} 2^{n}}{\binom{2n}{n}} 
p_{n} \left( - \frac{z}{2}; \frac{1}{4}, \, \frac{3}{4}, \, \frac{1}{4}, 
\, \frac{3}{4}, \, n \right).
\end{equation}
\noindent
This is discussed in \cite{hamdi-2010a} and \cite{koorwinder-1989a}.

The continuous Hahn polynomials satisfy the orthogonality condition
\begin{multline*}
\frac{1}{2 \pi} 
\int_{-\infty}^{\infty} \Gamma(a + \imath z) 
\Gamma(b + \imath z) 
\Gamma(c - \imath z) 
\Gamma(d - \imath z) 
p_{n}(z;a,b,c,d) p_{m}(z;a,b,c,d) \, dz   \\
 = \delta_{n,m} 
\frac{
\Gamma(n+a+d) 
\Gamma(n+b+c) 
\Gamma(n+a+c) 
\Gamma(n+b+d) }
{ n! (2n+a+b+c+d-1) \Gamma(n + a+ b+c+d-1)}.
\end{multline*}
\noindent
In the special case appearing here, using the identity 
$\begin{displaystyle} 
\left| \Gamma \left( \tfrac{1}{2} + \imath z \right) \right|^{2} = 
\pi \text{ sech}(\pi z)\end{displaystyle}$, the polynomials $P_{n}(z)$ are 
orthogonal on $\mathbb{R}$ with weight function $w(z) = \text{sech}(\pi z)$
and generating function 
\begin{equation}
\sum_{n=0}^{\infty} \frac{(2t)^{n}}{n!} P_{n}(z) = 
\frac{\exp( 2 z) \, \text{ arctan } t}{\sqrt{1+t^{2}}}.
\end{equation}
\end{note}

\begin{example}
\label{bateman}
The second example corresponds to the coefficients 
$a_{n,k} = \binom{2n}{n}^{-1} \binom{n}{k}^{2}$. Theorem 
\ref{formula-pn} now gives 
\begin{equation}
P_{n}(z) = \frac{\imath^{n} (n!)^{3}}{(2n)!} 
\sum_{k=0}^{n} \binom{n}{k}^{2} \binom{-\imath z - \tfrac{1}{2} +k}{n}.
\label{poly-sqbin}
\end{equation}
\noindent
To identify this class of polynomials it is convenient to convert them 
to hypergeometric form to produce 
\begin{equation}
P_{n}(z) = \frac{\imath^{n} (n!)^{3}}{(2n)!} 
\pFq32{-n,-n,\tfrac{1}{2} - \imath z}{1,\tfrac{1}{2}-\imath z -n}{1},
\label{poly-hyper0}
\end{equation}
\noindent
as remarked in \cite{bender-1988a}.

An alternative form of these polynomials is given next. 

\begin{lemma}
\label{poly-bateman}
The polynomials $P_{n}(z)$ in \eqref{poly-hyper0} are given by 
\begin{equation}
P_{n}(z) = \frac{(-\imath)^{n} (n!)^{3}}{(2n)!}
\pFq32{-n,n+1,\tfrac{1}{2} - \imath z}{1,1}{1}.
\end{equation}
\end{lemma}
\begin{proof}
Both sides are polynomials in $z$, so it suffices to verify the identity 
\begin{equation}
\pFq32{-n,n+1,m+1}{1,1}{1} = (-1)^{n} \binom{m}{n} 
\pFq32{-n,-n,m+1}{1,m-n+1}{1},
\end{equation}
\noindent
obtained from the value $z = \left( m + \tfrac{1}{2} \right) \imath$.  
This may be written in the equivalent form 
\begin{equation}
\sum_{k=0}^{n} (-1)^k \binom{n+k}{k} \binom{m+k}{k} \binom{n}{k} 
= (-1)^n \sum_{k=0}^{n} \binom{n}{k}^{2} \binom{m+k}{n}.
\end{equation}
\noindent
A direct calculation shows that the identity holds for $n=0, \, 1$ and the 
WZ-method shows that both sides satisfy the recurrence
\begin{equation}
(n+2)^{2}u(n+2,m) + (2m+1)(2n+3)u(n+1,m) - (n+1)^{2}u(n,m) = 0.
\end{equation}
\end{proof}
\end{example}

\begin{note}
The polynomials in Lemma \ref{poly-bateman}
can be expressed in terms of the 
\textit{Bateman polynomials} (see Section $18.19$ in 
 \cite{olver-2010a}): 
\begin{equation}
F_{n}(z) = \pFq32{-n,n+1,\frac{1+z}{2}}{1,1}{1} \text{ as }
P_{n}(z) = \frac{(- \imath)^{n} (n!)^{3}}{(2n)!} F_{n}(-2 \imath z).
\end{equation}
\noindent
It is curious that these seem to be the only powers of binomial coefficients
that give orthogonal polynomials. The powers up to $50000$ have been 
excluded using the first condition in Lemma \ref{ortho-1a}.
\end{note}

\section{Pyramid for the monomial $z^{n}$}
\label{sec-monomial}

For enumerative purposes, denote 
$[n]:=\{1,\dots,n\}$ and $[-n,n]:=\{\pm1,\dots,\pm n\}$. The \it symmetric 
group \rm $S_n$ is the set of
permutations of $[n]$. The \it signed permutation \rm group $B_n$ (or, 
\it hyperoctahedral group\rm) is the permutations
$\pi$ of $[-n,n]$, provided  $\pi(-k)=-\pi(k)$. In the 
literature, $S_n$ (resp. $B_n$) is a Coxeter group type $A$
(resp. type $B$). A \it descent \rm is a position $k$ where the permutation 
value has decreased: $\pi(k-1)>\pi(k)$.
Convention: $\pi(0):=0$. The classical Eulerian sequence $A_{n,k}$ of 
type $A$ (resp. Eulerian sequence
$B_{n,k}$ of type $B$) enumerates $S_n$ (resp. $B_n$) with $k$ descents. The 
corresponding Eulerian polynomials
of type $A$ and type $B$ are defined, respectively, by 
$A_n(x)=\sum_{k=0}^nA_{n,k}x^k$ and $B_n(x)=\sum_{k=0}^nB_{n,k}x^k$.

The polynomials $A_n(x)$ and $B_n(x)$ have the following rational generating 
functions:
$$\frac{A_n(x)}{(1-x)^{n+1}}=\sum_{k=0}^{\infty}(k+1)^nx^k, \qquad 
\text{and} \qquad
\frac{B_n(x)}{(1-x)^{n+1}}=\sum_{k=0}^{\infty}(2k+1)^nx^k.$$

As a direct consequence, we find the connections
$$B_{2n}(x)=(1-x)^nA_n(x), \qquad\text{and}\qquad 
B_{2n+1}(x)=(1-x)^nA_{n+1}(x).$$

Several authors (see \cite{amdeberhan-2013c,beck-2012a,stanley-1988a}) 
considered some \it quantum \rm 
extensions of Eulerian polynomials.

The pyramid corresponding to the polynomials $P_{n}(z) = z^{n}$ is given 
by Theorem \ref{formula-an} as 
\begin{equation}
a_{n,k} = \frac{1}{n!2^{n}} \sum_{j=0}^{n-k} (-1)^{n-k-j} 
\binom{n+1}{n-k-j} (2j+1)^{n}.
\end{equation}
\noindent
Define $B_{n,k} = 2^{n}n!a_{n,k}$. The rest of the section shows that 
$B_{n,k}$ are the coefficients of type $B$ polynomials.

\begin{lemma}
The numbers $B_{n,k}$ are integers with bivariate generating function
\begin{equation}
\sum_{n=0}^{\infty} \left( \sum_{k=0}^{n} B_{n,k}x^{k} 
\right) \frac{z^{n}}{n!} = 
\frac{(1-x) e^{(1-x)z}}{1 - xe^{2z(1-x)}}.
\label{bnk-gf}
\end{equation}
\noindent
Letting $x \to 1$ shows that 
$\{ a_{n,k} \}$ are normalized; that is 
$a_{n,0} + a_{n,1} + \cdots + a_{n,n} = 1$.
\end{lemma}

Letting $x = - 1$ in \eqref{bnk-gf} gives a relation between the 
numbers $B_{n,k}$ and  the 
\textit{Euler numbers} $E_{n}$ defined by the generating function 
$\begin{displaystyle} 
\sum_{n=0}^{\infty} E_{n} \frac{z^{n}}{n!} = \frac{1}{\cosh z}. 
\end{displaystyle}$

\begin{corollary}
The numbers $B_{n,k}$ satisfy 
$\quad \begin{displaystyle} \sum_{k=0}^{n} (-1)^{k} B_{n,k} = 2^{n} E_{n}. 
\end{displaystyle}$
\end{corollary}

The numbers $B_{n,k}$ resemble the \textit{Eulerian numbers} $\left\langle 
\begin{matrix} n \\ k \end{matrix} \right\rangle$, the coefficients of the 
type $A$ polynomials, with 
bivariate generating function  
\begin{equation}
\sum_{n=0}^{\infty} \left( \sum_{k=0}^{n} 
\left\langle \begin{matrix} n \\ k \end{matrix} \right\rangle
x^{k} 
\right) \frac{z^{n}}{n!} = 
\frac{(1-x) e^{(1-x)z}}{1 - xe^{z(1-x)}}.
\label{enk-gf}
\end{equation}
\noindent 
The similarity extends to the explicit expressions 
\begin{equation}
B_{n,n-k} = (-1)^{n} \sum_{j=0}^{k} (-1)^{j} \binom{n+1}{2k+1-j} 
(2k+1-2j)^{n}
\end{equation}
\noindent
and 
\begin{equation}
\left\langle \begin{matrix} n \\ k \end{matrix} \right\rangle = 
\sum_{j=0}^{k} (-1)^{j} \binom{n+1}{j} 
(k+1-j)^{n}.
\end{equation}

\medskip

\noindent
\textbf{Acknowledgments}. The fourth author acknowledges the partial 
support of NSF-DMS 1112656. The third author is a post-doctoral fellow 
funded in part by the same grant.

\end{document}